\documentclass{article}
\usepackage{amsmath,amssymb,theorem}
\usepackage{graphicx}
\usepackage{verbatim}

\usepackage{color}
\usepackage{makeidx}

\theorembodyfont{\upshape}

 \makeatletter
 \@addtoreset{equation}{section}
 \makeatother

 \newtheorem{ittheorem}{Theorem}
 \newtheorem{itlemma}{Lemma}
 \newtheorem{itproposition}{Proposition}
 \newtheorem{itdefinition}{Definition}

 \newtheorem{itremark}{Remark}
 \newtheorem{itclaim}{Claim}
 \newtheorem{itcorollary}{\bf Corollary}

 \newenvironment{theorem}{\addtocounter{equation}{1}
 \begin{ittheorem}}{\end{ittheorem}}

 \newenvironment{lemma}{\addtocounter{equation}{1}
 \begin{itlemma}}{\end{itlemma}}

 \newenvironment{proposition}{\addtocounter{equation}{1}
 \begin{itproposition}}{\end{itproposition}}

 \newenvironment{definition}{\addtocounter{equation}{1}
 \begin{itdefinition}}{\end{itdefinition}}

 \newenvironment{remark}{\addtocounter{equation}{1}
 \begin{itremark}}{\end{itremark}}

 \newenvironment{claim}{\addtocounter{equation}{1}
 \begin{itclaim}}{\end{itclaim}}

 \newenvironment{proof}{\noindent {\bf Proof.\,}
 }{\hspace*{\fill}$\qed$\medskip}

 \newenvironment{corollary}{\addtocounter{equation}{1}
 \begin{itcorollary}}{\end{itcorollary}}

 \newcommand{\be}[1]{\begin{eqnarray*}\label{#1}}
 \newcommand{\ee}{\end{eqnarray*}}

 \newcommand{\bl}[1]{\begin{lemma}\label{#1}}
 \newcommand{\el}{\end{lemma}}

 \newcommand{\br}[1]{\begin{remark}\label{#1}}
 \newcommand{\er}{\end{remark}}

 \newcommand{\bt}[1]{\begin{theorem}\label{#1}}
 \newcommand{\et}{\end{theorem}}

 \newcommand{\bd}[1]{\begin{definition}\label{#1}}
 \newcommand{\ed}{\end{definition}}

 \newcommand{\bcl}[1]{\begin{claim}\label{#1}}
 \newcommand{\ecl}{\end{claim}}

 \newcommand{\bp}[1]{\begin{proposition}\label{#1}}
 \newcommand{\ep}{\end{proposition}}

 \newcommand{\bc}[1]{\begin{corollary}\label{#1}}
 \newcommand{\ec}{\end{corollary}}

 \newcommand{\bpr}{\begin{proof}}
 \newcommand{\epr}{\end{proof}}

 \newcommand{\bi}{\begin{itemize}}
 \newcommand{\ei}{\end{itemize}}

 \newcommand{\ben}{\begin{enumerate}}
 \newcommand{\een}{\end{enumerate}}

\def\zl{\{\,0,\dots,\ell\,\}}

\def\uro{\smash{{U}^{\!\!\!\!\raise5pt\hbox{$\scriptstyle o$}}}}

\def\bp{{\overline{p}}}

\def\bp{{\overline{p}}}
\def\exa{\exp(-a)}
\def\exa{e^{-a}}

 \def \ba {\begin{array}}
 \def \ea {\end{array}}

 \def \qed {{\heartsuit\hfill}}
 
 \def \R {{\mathbb R}}
 
 \def \N {{\mathbb N}}

 \def \cS {{\cal S}}

\def\fs{\mathfrak S} 
\def\wE{\smash{\widetilde E}}
\def\wP{\smash{\widetilde P}}

 \def \cC {{\cal C}}
 
 \def \cQ {{\cal Q}}

\def \qed {{\square\hfill}}

 \def \s {{\sigma}}

\newcommand{\updown}[2]{\genfrac{\lbrace}{\rbrace}{0pt}{}{#1}{#2}}

  \def\cC{{\cal C}} 
   \def\cH{{\cal H}}

\def\cQ{{\cal Q}}  \def\cS{{\cal S}}

\def \qed {{\square\hfill}}

\def\R{{\mathbb R}}

\def\N{{\mathbb N}}

\def\eqref#1{(\ref{#1})}

\makeindex

\begin{document}
\allowdisplaybreaks[4]

\title{The quasispecies distribution}

 \author{
Rapha\"el Cerf and Joseba Dalmau
\\
DMA, 
{\'E}cole Normale Sup\'erieure\\
}

\maketitle



\begin{abstract}
\noindent
The quasispecies model was introduced in 1971 by Manfred Eigen to 
discuss
the first stages of life on Earth.
It provides an appealing mathematical framework to study the evolution
of populations in biology, for instance viruses.
We present briefly the model and we focus on its stationary solutions.
These formulae have a surprisingly rich combinatorial structure,
involving for instance the Eulerian and Stirling numbers,
as well as the up--down coefficients of permutations.
\end{abstract}



\section{Introduction}
The very concept of quasispecies is actively debated in theoretical biology. 
Loosely speaking,
a quasispecies is a group of individuals which are closely related to each other. 
At the genetic level, it is a model for a cloud of mutants around a well fitted
genotype, called the wild type or the master sequence.
Some biologists
argue that natural evolution operates on quasispecies rather than on single individuals.
Ideas coming from the quasispecies theory have been successfully applied to model
populations of viruses. 
Viruses 
have simple genomes
which can be analyzed with modern sequencing techniques.
Moreover they mutate very fast, thereby giving rise to complex quasispecies.
Some
medical strategies to prevent the development of viruses, like the HIV virus, are based
on the quasispecies model.
It is therefore crucial to improve our mathematical understanding of the
quasispecies model, in order to derive quantitative results which can be
confronted with experimental data.
In this text, we shall present briefly the quasispecies model of Eigen and
we shall study its stationary solutions. In doing so, we will embark on an
enriching journey around a wealth of mathematical tools:
Perron Frobenius theory, 
the polylogarithm or Jonqui\`ere's function,
Eulerian and Stirling numbers,
the up--down coefficients of permutations, the Poisson random walk
and traps on random permutations.

\section{The quasispecies model}


Manfred Eigen introduced the quasispecies model
in his celebrated article from 1971
about the first stages of life on Earth~\cite{EI1}.
Most presumably, the first living creatures were long macromolecules.
Eigen suggested that, at the macroscopic level, their evolution could be adequately
described by a collection of chemical reactions.
The main forces driving this evolution are selection and mutation.
Accordingly, the chemical reactions model the replication or the degradation
of each type of macromolecule.
Moreover the replication process is subject to errors caused by mutations.
Each type of macromolecule is classified according to its genotype. We denote by $E$ the
set of the possible genotypes.
The speed of reproduction of a macromolecule is a function of its genotype 
and it is given by a fitness function
$f:E\to{\mathbb R}^+$.
Finally, the probability that a macromolecule with genotype~$u$ mutates into
a macromolecule with genotype~$v$ is denoted by $M(u,v)$.
The concentration $x(v)$ of the genotype $v\in E$ 
evolves according to the differential equation
$$\frac{d}{dt}x_t(v)\,=\,
\sum_{u\in E}x_t(u)f(u)M(u,v)
-x_t(v)\sum_{u\in E}x_t(u)f(u)\,.$$
The first term
accounts for the production
of individuals having genotype $v$,
production due to erroneous replication of other genotypes
as well as faithful replication of itself.
The negative term accounts 
for the loss of individuals having genotype $v$,
and keeps the total concentration of individuals constant.
We shall focus on the
stationary solutions of Eigen's system, that is, the solutions of the system
$$\forall u\in E\qquad
x(u)\,\sum_{v\in E}x(v)f(v)\,=\,
\,\sum_{v\in E}x(v)f(v)M(v,u)\,
\qquad\qquad(\cal S)$$
subject to the constraint
$$\forall u\in E\qquad x(u)\geq 0\,,\qquad
\sum_{u\in E}x(u)\,=\,1\,.  \qquad\qquad(\cal C)$$
\section{Perron--Frobenius}\label{PF}
Suppose that $(x(u))_{u\in E}$ is 
a solution to~$(\cal S)$ which satisfies 
$(\cal C)$.
The mean fitness
$\sum_{v\in E}x(v)f(v)$
is then an eigenvalue 
of the matrix 
$$fM\,=\,(f(u)M(u,v))_{u,v\in  E}\,,$$
and 
$(x(u))_{u\in E}$ is an associated eigenvector,
whose components are non--negative.
There is a well--known framework where this problem has a
satisfactory and simple answer, given by the 
famous 
Perron--Frobenius theorem \cite{SEN}.
This theorem can be applied to any finite matrix with positive coefficients.
Thus we consider the following hypothesis. 
\medskip

\noindent
{\bf Hypothesis $(\cH)$.}
We suppose that the genotype space $E$ is finite, that
the fitness function $f$ is positive and that all the coefficients of the
mutation matrix $M$ are positive.
\medskip

\noindent
Under hypothesis $(\cH)$, 
we can apply the 
Perron--Frobenius theorem to
the matrix $fM$.
Let $\lambda$ be its Perron--Frobenius eigenvalue. 
The corresponding eigenspace has dimension~$1$ and it contains
an eigenvector associated to~$\lambda$ whose components
are non--negative. Moreover any eigenvector of $fM$ whose components are all
non--negative is associated to the eigenvalue~$\lambda$.
Therefore
the system~$(\cal S)$ admits a uni\-que solution 
satisfying the constraint $(\cal C)$. This solution is the
eigenvector 
$(x(u))_{u\in E}$ of the matrix $fM$,
associated to the Perron--Frobenius eigenvalue~$\lambda$,
which satisfies in addition
$$\lambda\,=\,\sum_{v\in E}x(v)f(v)\,.$$
\noindent

\section{Genotypes and mutations}
Ideally, we would like to have explicit formulae for $\lambda$ and $x$ in terms of $f$ and $M$.
There is little hope of obtaining such explicit formulae in the general case. Therefore,
we focus on a particular choice of
the set of genotypes $E$ 
and of the mutation matrix $M$. 
Both for practical and historical reasons,
we make the same choice as Eigen did.

\medskip
\noindent
{\bf Genotypes.}
We consider the different genotypes to be sequences of length $\ell\geq1$ over the alphabet $\lbrace\,0,1\,\rbrace$.
The space $E=\lbrace\,0,1\,\rbrace^\ell$
is often referred to as the $\ell$--dimensional hypercube. 
The hypercube is endowed with a natural
distance, called the Hamming distance,
which counts the number of
different digits between two different sequences: 
$$\forall\,u,v\in\lbrace\,0,1\,\rbrace^\ell\qquad
d(u,v)\,=\,\text{card}\,\big\lbrace\,
1\leq i\leq \ell:
u(i)\neq v(i)
\,\big\rbrace\,.$$
\noindent
{\bf Mutations.}
We suppose that mutations happen independently over each site of the sequence,
with probability $q\in\,]0,1[\,$.
For $u,v\in\lbrace\,0,1\,\rbrace^\ell$,
the mutation probability $M(u,v)$ is
thus given by
$$M(u,v)\,=\,q^{d(u,v)}(1-q)^{\ell-d(u,v)}\,.$$
\noindent
We have not specified the fitness function yet.
Let us consider first the simplest possible scenario, a constant fitness function:
$f(u)=c>0$ for all $u\in \lbrace\,0,1\,\rbrace^\ell$. 
When the fitness function is constant, there is no selection among different genotypes, and we say that the population is selectively neutral. 
Under the constraint $(\cC)$, since $f$ is constant,
$$\lambda\,=\,
\sum_{v\in E}x(v)f(v)\,=\,c\,.$$
With our choice of the mutation scheme, the matrix $M$ is symmetric, 
thanks to the symmetry of the Hamming distance. The matrix $M$ is also stochastic,
that is, each row of the matrix adds up to 1.
It is thus doubly stochastic, that is,
each column of the matrix adds up to 1 too.
We conclude that,
for a constant fitness function, the unique solution of $(\cS)$
satisfying the constraint $(\cC)$ is given by
$$x(u)\,=\,\frac{1}{|E|}\,=\,\frac{1}{2^\ell}\,,\qquad
u\in\lbrace\,0,1\,\rbrace^\ell\,.$$
However,
adaptive neutrality is seldom found in biological populations.
We thus embark on a quest for explicit formulae involving more complex fitness functions.
\section{Sharp peak landscape}
The simplest non neutral fitness function
which comes to mind
is the sharp peak:
there is a privileged genotype,
$w^*\in\lbrace\,0,1\,\rbrace^\ell$,
referred to as the master sequence,
which has a higher fitness than the rest.
Let $\sigma>1$ and let 
the fitness function $f$ be given by
$$\forall\,u\in\lbrace\,0,1\,\rbrace^\ell\qquad
f(u)\,=\,\begin{cases}
\quad\sigma&\quad\text{if}\quad u=w^*\,,\\
\quad1&\quad\text{if}\quad u\neq w^*\,.
\end{cases}$$
This is the fitness function that Eigen studied in detail in his article~\cite{EI1}.
One of the main advantages of working with the sharp peak is that we can break the space of genotypes into Hamming classes.
For $k\in\zl$,
the Hamming class $k$,
denoted by $\cH_k$,
is the subset of $\lbrace\,0,1\,\rbrace^\ell$
containing all the genotypes
that are at Hamming distance $k$ 
from the master sequence.
%
Let us define the function
$f_H:\zl\to\R^+$ by
$$\forall\,k\in\zl\qquad
f_H(k)\,=\,\begin{cases}
\quad\sigma&\quad\text{if}\quad k=0\,,\\
\quad1&\quad\text{if}\quad k>0\,.
\end{cases}$$
For each $k$,
the value $f_H(k)$
is the fitness common to all the genotypes in the Hamming class $k$.
As the next lemma shows,
the mutation probabilities can also be lumped over Hamming classes.
Let $b,c\in\zl$ and let $X,Y$
be independent random variables with binomial distributions
$X\sim\text{Bin}(b,q)$, 
${Y\sim\text{Bin}(\ell-b,q)}$
and define
$$M_H(b,c)\,=\,P\big(
b-X+Y=c
\big)\,.$$
\begin{lemma}
Let $b,c\in\zl$.
For any genotype $u$ in the Hamming class $b$, we have
$$\sum_{v\in\cH_c}M(u,v)\,=\,M_H(b,c)\,.$$
\end{lemma}
\begin{proof}
The quantity $\sum_{v\in\cH_c}M(u,v)$
is the probability of $u$ 
ending up in the class $c$ after mutation.
We call digits in a given genotype
correct or incorrect depending on whether 
they coincide with the master sequence or not.
Since $u$ is in the Hamming class $b$,
it has $b$ incorrect digits and $\ell-b$ correct ones.
Each digit changes state according to a Bernoulli random variable of parameter $q$.
Therefore,
the law of creating correct digits from the incorrect ones is $\text{Bin}(b,q)$.
Likewise, the law of creating incorrect digits from the correct ones is $\text{Bin}(\ell-b,q)$.
Noting that these binomial laws are independent of the placement of the correct and incorrect digits (and therefore of each other), we get the desired result.
\end{proof}

\noindent
Let $k\in\zl$.
Adding up the equations
of the system $(\cS)$ for $u\in\cH_k$
we get
$$
\sum_{u\in\cH_k}
x(u)\sum_{0\leq h\leq\ell}
\,\sum_{v\in\cH_h}x(v)f(v)\,=\,
\sum_{0\leq h\leq \ell}
\,\sum_{v\in\cH_h}x(v)f(v)
\sum_{u\in\cH_k}M(v,u)\,.
$$
We set 
$y(k)=\sum_{u\in\cH_k}x(u)$.
In view of the above remarks, we obtain 
the system
$$
y(k)\sum_{0\leq h\leq \ell}
y(h)f_H(h)\,=\,
\sum_{0\leq h\leq \ell}y(h)f_H(h)M_H(h,k)\,,\quad
0\leq k\leq\ell
\,.$$
The number of equations has been reduced from $2^\ell$ to $\ell+1$.
Moreover the new system still has the same form as $(\cS)$, and therefore all the considerations of section~\ref{PF} still hold for the new system.
Under the constraint $(\cC)$, the mean fitness
might be rewritten as
$$\sum_{0\leq h\leq \ell}y(h)f_H(h)\,=\,
(\sigma-1)y(0)+1\,.$$
The above system becomes then
$$
y(k)\big(
(\sigma-1)y(0)+1
\big)\,=\,
\sum_{0\leq h\leq \ell}y(h)f_H(h)M_H(h,k)\,,\quad
0\leq k\leq\ell
\,.$$
\section{Long chain regime}
Although this system of equations is much simpler than the initial one,
explicit formulae for $y$
are still out of hand.
In order to get simple and useful formulae,
we consider the asymptotic regime
$$\ell\to+\infty\qquad\qquad
q\to 0\qquad\qquad
\ell q\to a\in\,]0,+\infty[\,\,.$$
This asymptotic regime,
already considered by Eigen,
 arises naturally 
when modeling a population of individuals with a very long genome, in which the mean number of observed mutations per individual per generation is $a$. 
\begin{lemma}
Let $b,c\geq 0$. 
The mutation probability $M_H(b,c)$ satisfies
$$\lim_{\genfrac{}{}{0pt}{1}{\ell\to\infty,\,q\to 0}{\ell q\to a}}
\,M_H(b,c)\,=\,\begin{cases}
\quad\displaystyle\exa\frac{a^{c-b}}{(c-b)!}\quad&\text{if}\quad c\geq b\,,\\
\quad 0\quad&\text{if}\quad c<b\,.
\end{cases}$$
\end{lemma}
\begin{proof}
Recall that if $X\sim\text{Bin}(b,q)$
and $Y\sim\text{Bin}(\ell-b,q)$
are independent random variables,
then
$$M_H(b,c)\,=\,P(-X+Y=c-b)\,.$$
Since $b$ is fixed,
the law $\text{Bin}(b,q)$ converges
to a Dirac mass at $0$,
and the law $\text{Bin}(\ell-b,q)$
converges to a Poisson law of parameter $a$.
The formula appearing in the lemma
is precisely the probability of a Poisson random variable of parameter $a$ being equal to $c-b$.
\end{proof}

\medskip
\noindent
In view of this lemma, 
passing to the limit in the finite system, we obtain
the infinite system of equations
$$
y(k)\big(
(\sigma-1)y(0)+1
\big)\,=\,\sum_{0\leq h\leq k}
y(h)f_H(h)e^{-a}\frac{a^{k-h}}{(k-h)!}\,,\quad k\geq0.\quad
(\cS_{sp})$$
Let's take a look at the equation for $k=0$ first:
$$y(0)\big(
(\sigma-1)y(0)+1
\big)\,=\,y(0)\sigma e^{-a}\,.$$
The only two solutions to this equation are
$$y(0)\,=\,0\qquad\qquad\text{and}\qquad\qquad
y(0)\,=\,\frac{\sigma e^{-a}-1}{\sigma-1}\,.$$
On one hand, if $y(0)=0$, it can be seen by induction that $y$ is identically $0$, so this solution does not satisfy the constraint $(\cC)$.
On the other hand, the second solution for $y(0)$ is positive if and only if 
$\sigma e^{-a}>1$. Let us suppose that
$\sigma e^{-a}>1$, for we can only expect to find a solution satisfying 
the constraint $(\cC)$ in this case, and let us solve the recurrence relation defined by
$(\cS_{sp})$, with initial condition $y(0)=(\sigma e^{-a}-1)/(\sigma-1)$.
Replacing $y(0)$ on the left hand side 
of $(\cS_{sp})$ and dividing by $e^{-a}$
on both sides, the recurrence relation becomes
$$y(k)\sigma\,=\,y(0)\sigma\frac{a^k}{k!}
+\sum_{1\leq h\leq k}y(h)\frac{a^{k-h}}{(k-h)!}\,,\qquad
k\geq1\,.
$$
\section{The distribution of the quasispecies}
We choose to solve the recurrence relation by the method of generating functions (a beautiful account of this method can be found in chapter~7 of~\cite{GKP}). Set
$$g(X)\,=\,\sum_{k\geq 0}y(k)X^k\,.$$
Using the recurrence relation, we have
\begin{multline*}
g(X)e^{aX}\,=\,
\sum_{k\geq 0}\sum_{h=0}^k
y(h)\frac{a^{k-h}}{(k-h)!}X^k\\
=\,\sum_{k\geq 0}\Big(
y(k)\sigma-y(0)(\sigma-1)\frac{a^k}{k!}
\Big)X^k\,=\,
\sigma g(X)-y(0)(\sigma-1)e^{aX}\,.
\end{multline*}
Replacing $y(0)$ by its value, we get
\begin{multline*}
g(X)\,=\,(\sigma e^{-a}-1)\frac{e^{aX}}{\sigma-e^{aX}}\,=\,
(\sigma e^{-a}-1)\sum_{h\geq 1}\Big(
\frac{e^{aX}}{\sigma}
\Big)^h\\
=\,(\sigma e^{-a}-1)\sum_{h\geq 1}\frac{1}{\sigma^h}\sum_{k\geq 0}\frac{(ah)^k}{k!}X^k
\,=\,(\sigma e^{-a}-1)\sum_{k\geq 0}\frac{a^k}{k!}
\sum_{h\geq 1}\frac{h^k}{\sigma^h}X^k\,.
\end{multline*}
We deduce from here that
$$\forall\,k\geq 0\qquad
y(k)\,=\,(\sigma e^{-a}-1)\frac{a^k}{k!}
\sum_{h\geq 1}\frac{h^k}{\sigma^h}\,.$$
Eigen described the quasispecies as a population of individuals having a positive concentration of the master sequence along with a cloud of mutants. 
We now have an explicit formula for the concentrations of the master sequence and the different mutant classes in Eigen's original quasispecies model.
\begin{definition}
Let $\sigma,a$ be such that $\sigma e^{-a}>1$.
We say that a random variable $X$ has the distribution of
the quasispecies of 
parameters $\sigma$ and $a$, 
and we write ${X\sim\cQ(\s,a)}$,
if
$$\forall k\geq 0\qquad P(X=k)\,=\,(\sigma e^{-a}-1)\frac{a^k}{k!}
\sum_{h\geq 1}\frac{h^k}{\sigma^h}\,.$$
\end{definition}
%
\begin{figure}
\hspace*{-1 em}
\includegraphics[trim=0.6cm 0.5cm 0.4cm 1.1cm, clip=true, scale=0.6]{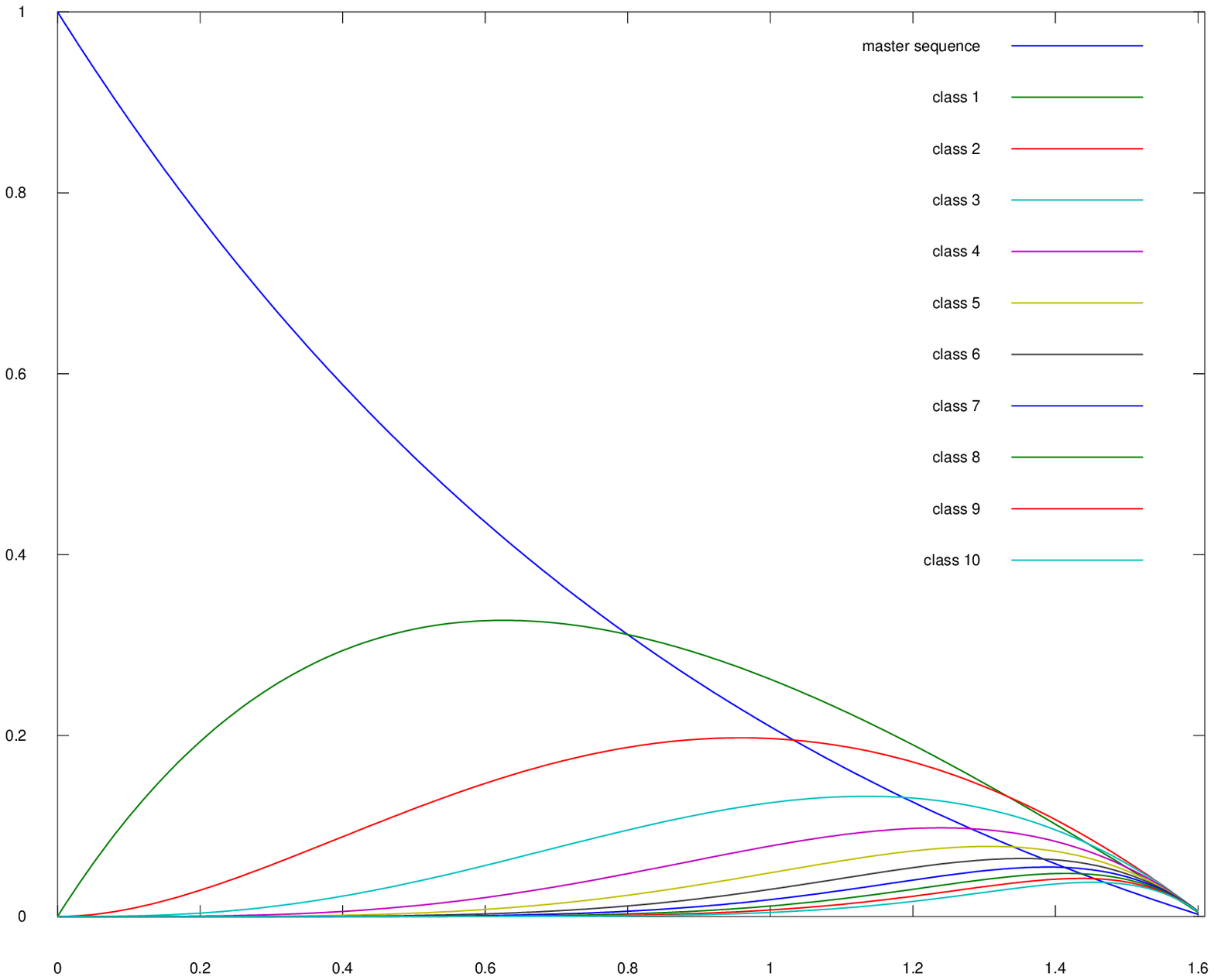}

Frequency of the MS
and the first 10 classes
as a function of $a$
for $\s=5$ .

%
%
 \hspace*{-1 em}
 \includegraphics[trim=0.6cm 0.5cm 0.4cm 0cm, clip=true, scale=0.6]{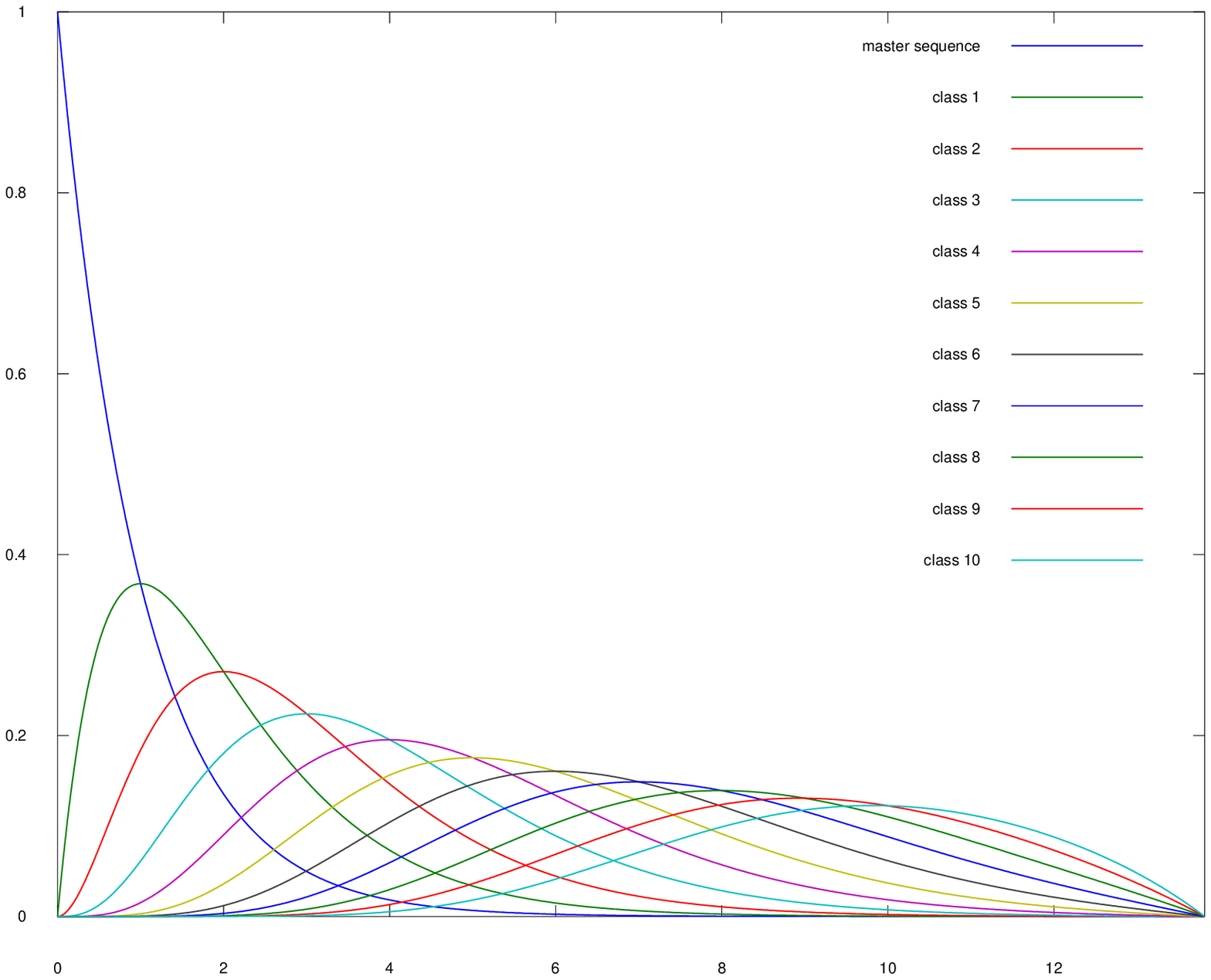}

Frequency of the MS
and the first 10 classes
as a function of $a$
for $\s=10^6$.
\end{figure}
The above formula 
is a genuine probability distribution, indeed all these numbers add up to one,
as can be seen by
replacing $X$ by $1$
in the equality
$$g(X)\,=\,(\sigma e^{-a}-1)\frac{e^{aX}}{\sigma-e^{aX}}\,.$$
The quasispecies distribution $\cQ(\s,a)$ can be expressed
in terms of the polylogarithm or Jonqui\`ere's function. Let $s,z\in\mathbb{C}$, with $|z|<1$.
The polylogarithm of order $s$ and argument $z$ is defined by
$$Li_{s}(z)\,=\,\sum_{h\geq 1}\frac{z^h}{h^s}\,.$$
In view of this definition,
$$\forall\, k\geq 0\qquad
y(k)\,=\,(\sigma e^{-a}-1)\frac{a^k}{k!}
Li_{-k}\Big(
\frac{1}{\sigma}
\Big)\,.$$
\section{Eulerian numbers}
We look next for an expression of $y(k)$
involving just a finite number of terms, instead of a series.
Let $s=1/\sigma$ and consider the well known identity
$$\sum_{h\geq 1}s^h\,=\,\frac{s}{1-s}\,.$$
We repeatedly derive and multiply by $s$
this equality, thus getting
\begin{align*}
\sum_{h\geq 1}hs^h\,&=\,\frac{s}{(1-s)^2}\,,\\
\sum_{h\geq 1}h^2s^h\,&=\,\frac{s}{(1-s)^3}(1+s)\,,\\
\sum_{h\geq 1}h^3s^h\,&=\,\frac{s}{(1-s)^4}(1+4s+s^2)\,,\\
\sum_{h\geq 1}h^4s^h\,&=\,\frac{s}{(1-s)^5}(1+11s+11s^2+s^3)\,.
\end{align*}
The numbers appearing on the right hand side
are the Eulerian numbers, and the polynomials
are the Eulerian polynomials.
\begin{definition}
For $0\leq h<k$,
the Eulerian number 
$\genfrac<>{0pt}{}{k}{h}$
is defined as the number of permutations of $\lbrace\,1,\dots,k\,\rbrace$
having exactly $k$ ascents,
that is, $k$ elements that are greater than the previous element in the permutation.
\end{definition}
The Eulerian numbers satisfy the
following identity:
$$\forall k\geq1\qquad
\sum_{h\geq 1}h^k s^h\,=\,
\frac{s}{(1-s)^{k+1}}
\sum_{h=0}^{k-1}
\genfrac<>{0pt}{}{k}{h}\,
s^h
\,.$$
Coming back to the variable $\sigma$, we get
$$\forall k\geq1\qquad
\sum_{h\geq 1}\frac{h^k}{\sigma^h}\,=\,
\frac{\sigma}{(\sigma-1)^{k+1}}
\sum_{h=0}^{k-1}
\genfrac<>{0pt}{}{k}{h}\,
\sigma^{k-h-1}
\,.$$
Using the classical identity
$\genfrac<>{0pt}{}{k}{h}\,=\,
\genfrac<>{0pt}{}{k}{k-1-h}$, and making the change of variable $h\to k-1-h$ in
the previous sum, we can express
the quantities $y(k)$ 
in terms of the Eulerian numbers:
$$\forall\,k\geq 1\qquad
y(k)\,=\,(\sigma e^{-a}-1)\frac{a^k}{k!}
\frac{\sigma}{(\sigma-1)^{k+1}}
\sum_{h=0}^{k-1}
\genfrac<>{0pt}{}{k}{h}\,
\sigma^{h}\,.$$
\section{Stirling numbers}
We have just seen that the concentration of class~$k$ in the
quasispecies distribution is a rational fraction in the variable~$\sigma$, with denominator
${(\sigma-1)^{k+1}}$ and numerator 
$(\sigma e^{-a}-1)\sigma$
times the $k$--th Eulerian polynomial.
Let us compute the partial fraction decomposition of this rational fraction. 
More precisely, we seek a sequence of real numbers $A_1,\dots,A_k$ such that
$$\forall\,k\geq 1\qquad
y(k)\,=\,(\sigma e^{-a}-1)\frac{a^k}{k!}
\frac{\sigma}{(\sigma-1)}
\sum_{h=1}^{k}
\frac{A_h}{(\sigma-1)^h}\,.$$
To find the values of the coefficients $A_h$, we write the Eulerian polynomial in terms
of the powers of $(\sigma-1)$:
\begin{align*}
\sum_{h=0}^{k-1}
\genfrac<>{0pt}{}{k}{h}\,
\sigma^{h}\,&=\,
\sum_{h=0}^{k-1}
\genfrac<>{0pt}{}{k}{h}\,
\sum_{j=0}^{h}
\genfrac(){0pt}{}{h}{j}(\sigma-1)^j\\
\,&=\,\sum_{j=0}^{k-1}
\Bigg(
\sum_{h=j}^{k-1}
\genfrac<>{0pt}{}{k}{h}\,
\genfrac(){0pt}{}{h}{j}\Bigg)
(\sigma-1)^j
\,.
\end{align*}
\begin{definition}
For $0\leq h\leq k$,
the Stirling number 
$\genfrac\{\}{0pt}{}{k}{h}$
is defined as the number of partitions of a set of cardinality $k$
into $h$ non empty subsets.
\end{definition}
The Stirling and Eulerian numbers are linked through the classical
identity
$$\sum_{h=j}^{k-1}
\genfrac<>{0pt}{}{k}{h}\,
\genfrac(){0pt}{}{h}{j}\,=\,(k-j)!\,
\genfrac\{\}{0pt}{}{k}{k-j}\,.$$
See for instance \cite{MT}, Proposition~$5.83$.
Reporting in the expression involving the Eulerian polynomial, and reindexing the sum,
we get
$$\forall\,k\geq 1\qquad
y(k)\,=\,(\sigma e^{-a}-1)\frac{a^k}{k!}
\frac{\sigma}{(\sigma-1)}
\sum_{h=1}^{k}
\frac{ h!
\genfrac\{\}{0pt}{}{k}{h}\,
}{(\sigma-1)^h}\,.$$

\section{Class--dependent fitness landscapes}
\label{cdf}
We have obtained explicit formulae
for the distribution of the quasispecies on the sharp peak landscape.
To get these formulae,
two ingredients have played
a key role: the Hamming classes and the asymptotic regime.
Yet, the strategy employed for the sharp
peak landscape still makes sense for a wider class of fitness functions,
namely, the fitness functions that only depend on the distance to the master sequence.
We consider thus the analogue of system $(\cS_{sp})$
for a general function $f:\mathbb{N}\to\R^+$:
$$
y(k)\sum_{h\geq0}y(h)f(h)\,=\,
\sum_{0\leq h\leq k}y(h)f(h)
e^{-a}\frac{a^{k-h}}{(k-h)!}\,,\quad k\geq 0\,.\qquad
(\cS_H)$$
We are only interested in the solutions of $\cS_H$ that satisfy constraint $(\cC)$
and such that $y(0)>0$.
For if $y(0)$ is a solution of $(\cS_H)$
with $y(0)=0$,
we can ignore the equation for $k=0$,
and the remaining system of equations falls into the form of $(\cS_H)$ again.
Thus, let us suppose that $y(0)>0$.
We look first at the equation for $k=0$:
$$y(0)\sum_{h\geq0}y(h)f(h)\,=\,y(0)f(0)e^{-a}\,.$$
Since we are assuming that $y(0)$ is positive,
the mean fitness, given by $\sum_{h\geq0}y(h)f(h)$,
must be equal to $f(0)e^{-a}$.
We make the change of variables $z(k)=y(k)/y(0)$,
we replace the mean fitness by $f(0)e^{-a}$ in $(\cS_H)$,
and we divide both sides by $e^{-a}$,
thus obtaining the recurrence relation
$$
z(k)f(0)\,=\,\sum_{0\leq h\leq k}z(h)f(h)\frac{a^{k-h}}{(k-h)!}\,,\qquad k\geq1\,,
$$
with initial condition $z(0)=1$. In order to get positive solutions, we make the following
hypothesis.

\medskip
\noindent
{\bf Hypothesis ($\mathcal{H}'$).}
We suppose that the fitness of the Hamming class $0$ is greater than the fitness of the other classes, i.e., $f(0)>f(k)$
for all $k\geq1$.

\medskip
\noindent
This hypothesis is coherent
with the Hamming class $0$ corresponding to the master sequence, which is 
the fittest genotype.
The method of generating functions cannot be implemented as easily as on the sharp peak
landscape.
However, it can be first
guessed and then shown by induction that, for all $k\geq1$,
$$z(k)\,=\,\frac{a^k}{k!}\frac{f(0)}{f(k)}
\!\!\!\!\!\!\sum_{\genfrac{}{}{0pt}{1}{1\leq h\leq k}{0=i_0<\cdots<i_h=k}}\!\!\!\!\!\!
\frac{k!}{(i_1-i_0)!\cdots(i_h-i_{h-1})!}
\prod_{1\leq t\leq h}\frac{f(i_t)}{f(0)-f(i_t)}\,.$$
\section{Up--down coefficients}
\label{upc}
If we apply the previous formula to the sharp peak landscape, we recover the formula
for the quasispecies involving the Stirling numbers. Indeed, in this case,
the last product
depends only on~$h$ (it is equal to $(\sigma-1)^h$) and the sum of the multinomial
coefficients is precisely equal to 
$h!
\genfrac\{\}{0pt}{}{k}{h}$.
There is yet another formula for the quantities $y(k)$,
which is the analogue of
the formula involving the Eulerian numbers in the case of the sharp peak landscape.
In order to present this formula,
we introduce the up--down numbers
or up--down coefficients.
Let $n\geq 2$,
and let $\sigma=(\sigma(1),\dots,\sigma(n))$
be a permutation of $1,\dots,n$.
The ascents and descents of $\sigma$
are codified by the Niven signature of $\sigma$,
that is, an array $(q_1,\dots,q_{n-1})\in\lbrace\,-1,+1\,\rbrace^{n-1}$
such that the product $q_i(\sigma(i+1)-\sigma(i))$
is positive for all $i$.
The up--down numbers, which we define next,
count the number of permutations sharing 
the same pattern of ascents and descents.
\begin{definition}
Let $n\geq2$ and let $I$ be a subset of $\{\,1,\dots, n-1\,\}$.
The up--down coefficient 
$\updown{n}{I}$
is defined as the number of permutations 
of $1,\dots,n$ having ascents in the positions $I$ and descents elsewhere.
In another words, it is the number of permutations of $1,\dots,n$ having for Niven's signature 
$$\forall\,i\in\lbrace\,1,\dots,n-1\,\rbrace\qquad
q_i\,=\,\begin{cases}
\quad +1\quad&\text{if}\quad i\in I\,,\\
\quad -1\quad&\text{if}\quad i\not\in I\,.\\
\end{cases}$$
\end{definition}
It turns out that
the quantities $z(k)$ can be expressed with the help 
of the up--down coefficients.
For all $k\geq 1$, we have
$$z(k)\,=\,\frac{a^k}{k!}
\bigg(
\prod_{1\leq j\leq k}\frac{f(0)}{f(0)-f(j)}
\bigg)
\sum_{I\subset\{\,1,\dots,k-1\,\}}
\bigg(
\updown{k}{I}
\prod_{i\in I}\frac{f(i)}{f(0)}
\bigg)
\,.$$
In the case of the sharp peak landscape, the last product depends only on the cardinality of $I$,
it is equal to $\sigma^{-|I|}$;
if we sum all the terms corresponding to subsets $I$ of cardinality $h$, we obtain
precisely the number of permutations of $1,\dots,k$ having $h$ ascents, which is equal
to the Eulerian number 
$\genfrac<>{0pt}{}{k}{h}$.

We obtained the above formula by writing explicitly the coefficients for small values of~$k$.
With the help of Sloane's on--line encyclopedia of integer sequences~\cite{Sloane}, we discovered that these
coefficients were the up--down coefficients.
Our first proof of the formula, done in \cite{CD},
relied on a difficult combinatorial identity due
to Carlitz \cite{CA}.
We present here a simpler more direct derivation. The strategy is to think of this formula
as a rational fraction in the variables $f(1),\dots,f(k)$ and to compute its partial
fraction decomposition, which turns out to be the formula given in
the previous section.
Thus we follow the inverse road that led us from the Eulerian numbers to the
Stirling numbers when we were playing with the quasispecies on the sharp peak
landscape.
Let us start. We define $K=\{\,1,\dots,k-1\,\}$, we
rewrite the above formula as
$$z(k)\,=\,\frac{a^k}{k!}
f(0)\bigg(
\prod_{1\leq j\leq k}\frac{1}{f(0)-f(j)}
\bigg)
\sum_{I\subset K}
\bigg(
\updown{k}{I}
{f(0)^{k-|I|-1}}
\prod_{i\in I}{f(i)}
\bigg)
\,$$
and we expand the power 
${f(0)^{k-|I|-1}}$ as
\begin{align*}
{f(0)^{k-|I|-1}}
\,&=\,
\prod_{j\in K\setminus I}{\big(f(0)-f(j)+f(j)\big)}\\
\,&=\,
\sum_{J\subset K\setminus I}
\bigg(\prod_{j\in J}{\big(f(0)-f(j)\big)}\bigg)
\bigg(
\prod_{j\in (K\setminus I)\setminus J}f(j)\bigg)\,.
\end{align*}
Reporting and simplifying the factors $(f(0)-f(j))$, we obtain
$$z(k)\,=\,\frac{a^k}{k!}
f(0)
\sum_{I\subset K}
\sum_{J\subset K\setminus I}
\bigg(
\prod_{j\in K\cup\{k\}\setminus J}\frac{1}{f(0)-f(j)}
\bigg)
\bigg(
\updown{k}{I}
\prod_{j\in K\setminus J}{f(j)}
\bigg)\,.
$$
We reindex the sum by setting $H=K\setminus J$ and we get
$$z(k)\,=\,\frac{a^k f(0)}{k!f(k)}
\sum_{H\subset K}
\bigg(
\prod_{j\in H\cup\{k\}}\frac{f(j)}{f(0)-f(j)}
\bigg)
\bigg(
\sum_{I\subset H}
\updown{k}{I}
\bigg)\,.
$$
Let us fix $H\subset K$, say
$H=\{\,i_1,\dots,i_{h-1}\,\}$, where $1\leq h\leq k$ and
$${i_0=1\leq i_1<\cdots<i_{h-1}<k=i_h\,,}$$
 and let us focus on the last sum
$\sum_{I\subset H} \updown{k}{I}$.
This sum is the number of permutations of $1,\dots, k$
whose ascents are located in the index set~$H$.
Let $B=\big(B_1,\dots,B_h\big)$ be an ordered partition of
$\{\,1,\dots, k\,\}$ in $h$ subsets such that
$$\forall j\in\{\,1,\dots ,h\,\}\qquad|B_{j}|=
i_{j}- i_{j-1}\,.$$
We list the elements of each set $B_j$ in decreasing order:
$$\forall j\in\{\,1,\dots ,h\,\}
\qquad B_{j}\,=\,\big(b_j(1),\dots,b_j(i_j-i_{j-1})\big)\,.$$
We concatenate these lists into a single sequence:
$$b_1(1),\dots,
b_1(i_1),b_2(1),\dots, b_2(i_2-i_1),\dots,
b_h(1),\dots,
b_h(i_h-i_{h-1})\,.$$
This sequence 
corresponds to a permutation of 
$1,\dots, k$. 
This construction defines a one to one correspondence between ordered 
partitions of
$\{\,1,\dots, k\,\}$ into $h$ subsets of respective
sizes $i_1,\dots,i_h-{i_{h-1}}$
and
the set of the permutations of $1,\dots, k$
whose ascents are located in the index set~$H$.
The number of these partitions 
(called $h$--sharing in the terminology of \cite{MT},
see definition~$1.17$ and proposition~$5.5$ therein)
is precisely the multinomial coefficient 
$\frac{k!}{(i_1-i_0)!\cdots(i_h-i_{h-1})!}$
and
we conclude that
$$\sum_{I\subset H}
\updown{k}{I}\,=\,
\frac{k!}{(i_1-i_0)!\cdots(i_h-i_{h-1})!}
\,.$$
In fact, this combinatorial identity and the above argument are the
starting point of 
Carlitz work \cite{CA}. The goal of Carlitz
was to invert this formula, i.e., to express the up--down coefficients
as sums of multinomial coefficients. 
Plugging this identity in the formula for $z(k)$, we are back to the
formula obtained by induction in section~\ref{cdf}.
\section{The Perron--Frobenius eigenvector}
The previous formulae for the quasispecies are a bit mysterious. They pop up
from inductions and combinatorial identities. We seek next probabilistic
representations of the quasispecies in order to shed some light on its structure.
With that goal in mind, we start again from the framework of 
hypothesis~$(\cH)$ in
section~\ref{PF},
i.e., a finite genotype space~$E$
and a situation where Eigen's system admits a unique 
stationary solution.
Let $(S_n)_{n\in\mathbb N}$ be the Markov chain on~$E$ with 
transition matrix~$M$.
We denote by $\lambda$ the Perron--Frobenius eigenvalue of $fM$.
For $u\in E$, 
we denote by $E_u$ the expectation for the Markov chain $(S_n)_{n\in\mathbb N}$
starting from $u$ 
and
we define
$$\tau_u\,=\,\inf\,\big\{\,n\geq 1: S_n=u\,\big\}\,.$$
\begin{theorem}
\label{profor}
Suppose that $(\cal H)$ holds.
Let $w$ be an arbitrary point of~$E$.
The unique solution 
to $(\cal S)$ which satisfies 
the constraint $(\cal C)$ is given by the formula
$$
\forall u\in E\qquad
x(u)\,=\,\frac{
\displaystyle
E_w\Bigg(\sum_{n=0}^{\tau_w-1}\Big(
1_{\{S_n=u\}}
\lambda^{-n}\prod_{k=0}^{n-1}f(S_k)
\Big)
\Bigg)
}
{
\displaystyle
E_w\Bigg(\sum_{n=0}^{\tau_w-1}\Big(
\lambda^{-n}\prod_{k=0}^{n-1}f(S_k)
\Big)
\Bigg)
}\,.
$$
\end{theorem}
Theorem~\ref{profor} can be verified directly by plugging the above formula
in the system~$(\cS)$.
By taking $w=u$ in the formula stated in theorem~\ref{profor},
we obtain the following corollary.
\begin{corollary}
\label{trofor}
Suppose that $(\cal H)$ holds.
The unique solution 
to $(\cal S)$ which satisfies 
the constraint $(\cal C)$ is given by the formula
$$
\forall u\in E\qquad
x(u)\,=\,{
\displaystyle
\frac{1}{
\displaystyle
E_u\Bigg(\sum_{n=0}^{\tau_u-1}\Big(
\lambda^{-n}\prod_{k=0}^{n-1}f(S_k)
\Big)
\Bigg)
}
}\,.
$$
\end{corollary}
This formula is a generalization of the classical formula for the invariant 
probability measure of a Markov chain. Indeed, in the particular case where
$f$ is constant equal to~$1$, then $\lambda=1$ as well, and the system~$(\cal S)$
reduces to
$$\forall u\in E\qquad
x(u)
\,=\,
\,\sum_{v\in E}x(v)M(v,u)\,,$$
while the formula in corollary~\ref{trofor} becomes the well--known formula
$$
\forall u\in E\qquad
x(u)\,=\,
\frac{1}{
\displaystyle
\displaystyle
E_u\big({\tau_u}\big)}\,.$$
The formula of
theorem~\ref{profor} is quite general, but it applies only to finite spaces.
However, it helps to construct plausible formulae for the infinite system.
The game consists in performing a formal passage to the limit and identifying
the relevant limiting probabilistic objects.
\section{The killed Poisson walk}
\label{kpw}
In the asymptotic regime, the eigenvalue~$\lambda$ converges towards 
$f(0) e^{-a}$
and the projection on the Hamming classes of the Markov chain
$(S_n)_{n\in\mathbb N}$
converges to a random walk on the integers whose steps are distributed according
to the Poisson law of parameter~$a$. This random walk is transient, therefore in the limit
the return time to~the master sequence is either $0$ or $\infty$, and 
the limiting formula reminds of the Poisson random walk killed at rate $1-1/\lambda$.
These considerations lead to the following construction for a plausible limit.
Let ${a}>0$ and
let 
$(X_n)_{n\geq 1}$
be a sequence of i.i.d. random variables distributed according to the Poisson
law of parameter~$a$:
$$\forall n\geq 1\quad\forall k\geq 0\qquad
P(X_n=k)\,=\,e^{-a}\frac{a^k}{k!}\,.$$
We consider the associated random walk on the non--negative integers, given by
$S_0=0$ and
$$\forall n\geq 1\qquad S_n\,=\,X_1+\cdots +X_n\,.$$
From now onwards, our goal is to obtain a probabilistic representation of
the quasispecies distribution in terms of the Poisson random walk killed 
at a random time. More precisely, we aim at constructing an integer--valued
random variable $\tau$ such that
the concentrations $\smash{\big(y(k)\big)}_{k\geq 0}$ of the Hamming classes
in the quasispecies are equal to
the mean empirical
distribution of the Poisson random walk between times $1$ and $\tau$,
that is,
$$\forall k\geq 0\qquad 
y(k)\,=\,
\frac{1}{E(\tau)}
E\bigg(\sum_{n=1}^\tau1_{\{S_n=k\}}\bigg)\,.
\qquad
(\diamondsuit)$$
We start this program on the sharp peak landscape. 
\begin{proposition}
Let $\sigma,a$ be such that $\sigma e^{-a}>1$.
Let $\tau$ be a random variable, which is independent of the Poisson random walk,
with geometric distribution of parameter 
$1-(\sigma e^{-a})^{-1}$.
With this choice of $\tau$, the probabilistic representation $(\diamondsuit)$ holds 
for the quasispecies distribution
${\cQ(\s,a)}$.
\end{proposition}
We recall that the geometric distribution of parameter 
$1-(\sigma e^{-a})^{-1}$ is 
$$\forall n\geq 1\qquad
P(\tau\geq n)\,=\,
\Big(
\frac{1}{\sigma e^{-a}}
\Big)^{n-1}
\,.$$
\begin{proof}
We compute the expectation by decomposing the sum according to the value of $\tau$:
$$\displaylines{
E\bigg(\sum_{n=1}^\tau 1_{\{S_n=k\}}\bigg)\,=\,
\sum_{t=1}^\infty
E\bigg(\sum_{n=1}^t 1_{\{S_n=k,\,\tau=t\}}\bigg)\hfil\cr
\hfil\,=\,
\sum_{t=1}^\infty
\sum_{n=1}^t 
P\big(
{S_n=k,\,\tau=t}\big)
\,=\,
\sum_{n=1}^\infty
\sum_{t=n}^\infty
P\big(
{S_n=k,\,\tau=t}\big)\,.
}$$
Now, the variables $S_n$ and $\tau$ are independent.
The distribution of $\tau$ is geometric, the distribution of $S_n$
is Poisson of parameter $na$ (it is a sum of $n$ independent Poisson
distributions of parameter~$a$). Thus the previous sums become
$$
\displaylines{
\sum_{n=1}^\infty
\sum_{t=n}^\infty
P\big(
{S_n=k\big)P\big(\tau=t}\big)
\,=\,
\sum_{n=1}^\infty
P\big(
{S_n=k\big)P\big(\tau\geq n}\big)
\hfil\cr
\hfil
\,=\,
\sum_{n=1}^\infty
e^{-na}\frac{(an)^k}{k!}
\Big(
\frac{1}{\sigma e^{-a}}
\Big)^{n-1}
\,=\,
\sigma
e^{-a}\frac{a^k}{k!}
\sum_{n=1}^\infty
\frac{n^k}{\sigma^n}
\,.
}
$$
Since $E(\tau)=\sigma e^{-a}/(
\sigma e^{-a}-1)$, we recover the quasispecies distribution on the sharp peak landscape.
\end{proof}

\noindent
The previous construction can be extended to a class--dependent fitness as follows.
Suppose that at
time $n$, 
the random walk $S_n$ is in state $i\geq 1$. We toss an independent coin
of parameter $e^af(i)/f(0)$ to decide whether the walk survives another unit of time or not.
More precisely, we define, 
for any $i,n\geq 0$,
$$
P\big(\tau\geq n+1\,\big|\,S_n=i,\,\tau\geq n\big)\,=\,
\begin{cases}
\quad 1\quad&\text{if}\quad i=0\,,\\
\quad 
\displaystyle
e^a\frac{f(i)}{f(0)}
\quad&\text{if}\quad i\geq 1\,.
\end{cases}$$
This defines a random time $\tau$ whose distribution
is a predictable function of the trajectory
of the Poisson walk
$(S_n)_{n\in\mathbb N}$, i.e., the event $\tau=n+1$ depends on $n$ independent
coins whose parameters are deterministic functions of the 
trajectory $S_0,\dots,S_n$ until time $n$. Of course, the definition of~$\tau$ makes sense
only when the following hypothesis holds.

\medskip
\noindent
{\bf Hypothesis ($\mathcal{H}''$).}
We suppose that 
$f(0)\geq e^a f(k)$
for all $k\geq1$.
\begin{proposition}
\label{kco}
Let $f$ be a fitness function satisfying hypothesis $(\cH'')$.
Let $(S_n)_{n\in\N}$ be the Poisson random walk and
let $\tau$ be the random time defined above.
With these choices, the probabilistic representation $(\diamondsuit)$ holds 
for the quasispecies distribution 
associated to~$f$. 
\end{proposition}
\begin{proof}
Let $k\geq 1$ and let us set
$$T_k\,=\,\inf\,\{\,n\geq 1:S_n=k\,\}\,.$$
We compute, with the help of a conditioning and the Markov property,
\begin{align*}
E\bigg(&\sum_{n=1}^\tau1_{\{S_n=k\}}\bigg)\,=\,
\sum_{t\geq 1}P\bigg(
\sum_{n=1}^\tau1_{\{S_n=k\}}\geq t\bigg)\cr
&\,=\,
\sum_{t\geq 1}
\sum_{m\geq 1}
P\big(T_k=m,\,S_m=\cdots=S_{m+t-1}=k,\,\tau\geq m+t-1\big)
\cr
&\,=\,
\sum_{t\geq 1}
\sum_{m\geq 1}
P\bigg(
\begin{matrix}
T_k=m\\
\tau\geq m
\end{matrix}
\bigg)\,
P\bigg(
\begin{matrix}
S_{m+1}=\cdots=S_{m+t-1}=k\\
\tau\geq m+t-1
\end{matrix}
\,\bigg|\,
\begin{matrix}
S_m=k\\
\tau\geq m
\end{matrix}
\bigg)\,
\cr
&\,=\,
\bigg(\sum_{m\geq 1}
P\bigg(
\begin{matrix}
T_k=m\\
\tau\geq m
\end{matrix}
\bigg)
\bigg)
\bigg(
\sum_{t\geq 1}
P\bigg(
\begin{matrix}
S_{1}=\cdots=S_{t-1}=k\\
\tau\geq t-1
\end{matrix}
\,\bigg|\,S_0=k\bigg)
\bigg)
\,.
\end{align*}
We deal separately with each sum.
First, we have
$$\displaylines{
\sum_{m\geq 1}
P\bigg(
\begin{matrix}
T_k=m\\
\tau\geq m
\end{matrix}
\bigg)\,=\,
\!\!\!\!\!\!\sum_{\genfrac{}{}{0pt}{1}{m\geq 1,\,
0=s_0\leq \cdots}{\leq s_{m-1}<s_m=k} }\!\!\!\!\!\!
P\big(S_1=s_1,\dots,S_m=s_m,\,\tau \geq m\big)\cr
\,=\,
\!\!\!\!\!\!\sum_{\genfrac{}{}{0pt}{1}{m\geq 1,\,
0=s_0\leq \cdots}{\leq s_{m-1}<s_m=k} }
\kern 7pt
\prod_{j=0}^{m-1}
\bigg(
\max\bigg(1,
\frac{f(s_j)}{f(0)e^{-a}}
\bigg)
e^{-a}
\frac{a^{s_{j+1}-s_{j}}}{(s_{j+1}-s_{j})!}
\bigg)
\,.
}$$
We reindex the sum according to the number~$h$ of distinct integers
in the trajectory 
$s_0\leq \cdots\leq s_{m-1}<s_m=k$ and we obtain
$$\displaylines{
\!\!\!\!\!\!\sum_{\genfrac{}{}{0pt}{1}
{1\leq h\leq k,\,t_0,\dots,t_{h-1}\geq 1}
{0< i_1< \cdots < i_{h-1}<i_h=k} 
}\!\!\!\!\!\!
\kern-7pt
e^{-a(t_0-1)}e^{-a}
\frac{a^{i_1}}{i_1!}
\prod_{j=1}^{h-1}
\Bigg(
\bigg(\frac{f(i_j)}{f(0)e^{-a}}\bigg)^{t_j}
e^{-a(t_j-1)}e^{-a}
\frac{a^{i_{j+1}-i_{j}}}{(i_{j+1}-i_{j})!}
\Bigg)\cr
\,=\,\!\!\!\!\!\!\sum_{\genfrac{}{}{0pt}{1}
{1\leq h\leq k,\,t_0,\dots,t_{h-1}\geq 1}
{0< i_1< \cdots < i_{h-1}<i_h=k} 
}\!\!\!\!\!\!
\kern-7pt
e^{-at_0}
\frac{ a^k }{i_1!}
\prod_{j=1}^{h-1}
\Bigg(
\bigg(\frac{f(i_j)}{f(0)}\bigg)^{t_j}
\frac{1}{(i_{j+1}-i_{j})!}
\Bigg)\cr
\,=\,\!\!\!\!\!\!\sum_{\genfrac{}{}{0pt}{1}
{1\leq h\leq k}
{0< i_1< \cdots < i_{h-1}<i_h=k} 
}\!\!\!\!\!\!
\kern-7pt
\frac{e^{-a}}{1-e^{-a}}
\frac{ a^k }{i_1!}
\prod_{j=1}^{h-1}
\Bigg(
\frac{f(i_j)}{f(0)-f(i_j)}
\frac{1}{(i_{j+1}-i_{j})!}
\Bigg)
\,.
}$$
Second, we have
$$\displaylines{
\sum_{t\geq 1}
P\bigg(
\begin{matrix}
S_{1}=\cdots=S_{t-1}=k\\
\tau\geq t-1
\end{matrix}
\,\bigg|\,S_0=k\bigg)
\,=\,\hfil
\cr
\hfil\sum_{t\geq 1}
\bigg(
e^{-a}
\frac{f(k)}{f(0)e^{-a}}
\bigg)^{t-1}
\,=\,
\frac{f(0)}{f(0)-f(k)}\,.
}$$
Collecting together the previous computations, we obtain
$$
\displaylines{
E\bigg(\sum_{n=1}^\tau1_{\{S_n=k\}}\bigg)\,=\,\hfill\cr
\frac{e^{-a}}{1-e^{-a}}
a^k
\frac{f(0)}{f(k)}
\!\!\!\!\!\!\sum_{\genfrac{}{}{0pt}{1}
{1\leq h\leq k}
{i_0=0< i_1< \cdots < i_{h-1}<i_h=k} 
}\!\!\!\!\!\!
\kern-7pt
\prod_{j=1}^{h}
\Bigg(
\frac{f(i_j)}{f(0)-f(i_j)}
\frac{1}{(i_{j}-i_{j-1})!}
\Bigg)
}
$$
and, up to a multiplicative constant,
we recognize the formula obtained in section~\ref{cdf} for the quantities $z(k)$.
\end{proof}
Our probabilistic construction provides the following intuitive picture for
the structure of the quasispecies. The evolution of the genotype along a lineage
is modelled by a Poisson random walk in the genotype space, starting from
the master sequence. Because of the presence of the master sequence in the population,
the lineages are bound to become extinct, after a random time which depends on their
fitness history. A lineage is more robust if it visits genotypes whose fitnesses
are close
to the fitness of the master sequence. The time $\tau$ models the survival time of a
lineage.
\section{Traps on ascents of random permutations}
We shall finally try to construct a probabilistic model corresponding to
the formula which involves the up--down coefficients.
The most natural random object associated to the up--down coefficients
is a random permutation. The problem is that, to realize the formula
giving $y(k)$, we should draw a random
permutation of $1,\dots,k$, and this for each $k\geq 1$.
Yet we wish to construct a random object whose distribution is given by the 
$y(k)$'s, so the construction should not depend on a fixed value of $k$.
Fortunately, there exists a smart way to embed the up--down coefficients associated
to random finite permutations into
an infinite random sequence. This construction is done by 
Oshanin and Voituriez \cite{OV} and it proceeds as follows.
Let 
$(Y_n)_{n\geq 0}$
be a sequence of i.i.d. random variables, with uniform distribution over $[0,1]$.
We declare that there is an ascent at index $i\geq 1$ if and only if $Y_{i+1}> Y_i$.
Let us fix an integer $k\geq 1$.
The indices of ascents until $k-1$ are exactly the positions
of the ascents of the permutation
$\fs$
of 
$1,\dots,k$ satisfying
$$Y_{\fs^{-1}(1)}<\dots<
Y_{\fs^{-1}(k)}\,.
$$
Since $Y_1,\dots,Y_k$ are i.i.d. uniform over $[0,1]$, then
the permutation
$\fs$ 
is uniformly distributed over
the permutations
of
$1,\dots,k$. Thus, for any fixed $k$, the distribution of the
indices of ascents until $k-1$
has the same distribution than the ascents of a permutation of
$1,\dots,k$ chosen uniformly at random.
Let us continue the construction of a random object associated
to the up--down formula of the quasispecies distribution.
We shall also need a random walk to account for the term $a^k$
in the formula.
Let $\sigma,a$ be such that $\sigma e^{-a}>1$.
Let 
$(X_n)_{n\geq 1}$
be a sequence of i.i.d. Bernoulli random variables such that
$$\forall n\geq 1\quad\forall k\geq 0\qquad
P(X_n=0)\,=\,{1-a}\,,\quad
P(X_n=1)\,=\,a
\,.$$
We consider the associated random walk on the non--negative integers, given by
$S_0=0$ and
$$\forall n\geq 1\qquad S_n\,=\,X_1+\cdots +X_n\,.$$
We define a random time $\tau$ 
as follows.
Suppose that, at
time $n$, 
the random walk $S_n$ is in state $i\geq 1$
and that it has survived $n$ units of time.
At time $n+1$, the walk tries to move to the point $S_{n+1}=S_n+X_n$.
It is killed at time $n+1$ according to the following rule.
We toss an independent coin
to decide whether the walk 
survives another unit of time or not.
In case there is no move, i.e., $X_n=0$, the parameter of the coin is
$f(i)/\big(f(0)(1-a)\big)$.
In case 
the move is a right step, i.e., $X_n=1$, and if in addition
$Y_{i+1}>Y_i$, the parameter of the coin is $f(i)/f(0)$.
In case 
the move is a right step and $Y_{i+1}<Y_i$, the parameter is $1$
and the survival is guaranteed.
More precisely, we define, 
for any $i,j,n\geq 0$,
$$
P\bigg(\tau\geq n+1\,\bigg|\,
\begin{matrix}
S_n=i\\
S_{n+1}=j\\
\tau\geq n
\end{matrix}
\bigg)\,=\,
\begin{cases}
\quad 1\quad&\text{if}\quad i=j=0\,,\\
\quad 
\displaystyle\frac{f(i)}{f(0)(1-a)}
\quad&\text{if}\quad i=j\geq 1\,,\\
\quad 
\displaystyle\frac{f(i)}{f(0)}
\quad&\text{if}\quad j=i+1,\,
Y_{i+1}>Y_i
\,,\\
\quad 
1
\quad&\text{if}\quad j=i+1,\,
Y_{i+1}<Y_i
\,.\\
\end{cases}$$
The definition of~$\tau$ makes sense
only when the following hypothesis holds.

\medskip
\noindent
{\bf Hypothesis ($\mathcal{H}'''$).}
We suppose that 
$f(0)(1-a) \geq f(k)$
for all $k\geq1$.
\begin{proposition}
Let $f$ be a fitness function satisfying hypothesis $(\cH''')$.
Let $(S_n)_{n\in\N}$ and
$\tau$ be the random walk and time defined above.
With these choices, the probabilistic representation $(\diamondsuit)$ holds 
for the quasispecies distribution 
associated to~$f$. 
\end{proposition}
\begin{proof}
The expectation in the formula is taken with respect to the variables
$X_n$ and $Y_n$. To prove the formula, we shall first fix the variables
$Y_n$ and take the expectation with respect to the variables $X_n$.
Let us denote by $\wE$ and $\wP$ the conditional expectation and probability
knowing the variables $(Y_n)_{n\geq 1}$.
Let $k\geq 1$. We set
$$T_k\,=\,\inf\,\{\,n\geq 1:S_n=k\,\}\,.$$
Exactly as in the proof of proposition~\ref{kco}, we have
$$\displaylines{
\wE\bigg(\sum_{n=1}^\tau1_{\{S_n=k\}}\bigg)\,=\,\hfill\cr
\hfil
\bigg(\sum_{m\geq 1}
\wP\bigg(
\begin{matrix}
T_k=m\\
\tau\geq m
\end{matrix}
\bigg)
\bigg)
\bigg(
\sum_{t\geq 1}
\wP\bigg(
\begin{matrix}
S_{1}=\cdots=S_{t-1}=k\\
\tau\geq t-1
\end{matrix}
\,\bigg|\,S_0=k\bigg)
\bigg)
\,.
}$$
We deal separately with each sum.
First, we have
$$\displaylines{
\sum_{m\geq 1}
\wP\bigg(
\begin{matrix}
T_k=m\\
\tau\geq m
\end{matrix}
\bigg)\,=\,
\sum_{t_0,\dots, t_{k-1}\geq 1}
\Bigg(
\big({1-a}\big)^{t_0-1}
{a}\hfill
\cr
\prod_{i=1}^{k-1}
\Bigg(
\big({1-a}\big)^{t_i-1}
\bigg(
\displaystyle\frac{f(i)}{f(0)(1-a)}
\bigg)^{t_i-1}
a
\bigg(
\displaystyle\frac{f(i)}{f(0)}1_{Y_{i+1}>Y_{i}}+1_{Y_{i+1}<Y_{i}}\bigg)
\Bigg)
\Bigg)\cr
\,=\,
a^{k-1}
\prod_{i=1}^{k-1}
\Bigg(
\frac{f(0)}{f(0)-f(i)}
\bigg(
\displaystyle\frac{f(i)}{f(0)}1_{Y_{i+1}>Y_{i}}+1_{Y_{i+1}<Y_{i}}\bigg)
\Bigg)
\,.
}$$
Second, we have
$$\displaylines{
\sum_{t\geq 1}
\wP\bigg(
\begin{matrix}
S_{1}=\cdots=S_{t-1}=k\\
\tau\geq t-1
\end{matrix}
\,\bigg|\,S_0=k\bigg)
\,=\,\hfil
\cr
\hfil\sum_{t\geq 1}
\big(
{1-a}
\big)^{t-1}
\bigg(
\displaystyle\frac{f(k)}{f(0)(1-a)}
\bigg)^{t-1}
\,=\,
\frac{f(0)}{f(0)-f(k)}\,.
}$$
Collecting together the previous computations, we obtain
$$
\displaylines{
\wE\bigg(\sum_{n=1}^\tau1_{\{S_n=k\}}\bigg)\,=\,\hfill\cr
a^{k-1}
\Bigg(
\prod_{i=1}^{k}
\frac{f(0)}{f(0)-f(i)}
\Bigg)
\Bigg(
\prod_{j=1}^{k}
\bigg(
\displaystyle\frac{f(j)}{f(0)}1_{Y_{j+1}>Y_{j}}+1_{Y_{j+1}<Y_{j}}\bigg)
\Bigg)\,.
}
$$
It remains to average over the 
variables $(Y_n)_{n\geq 1}$. We have
$$
\displaylines{
E\Bigg(\prod_{j=1}^{k}
\bigg(
\displaystyle\frac{f(j)}{f(0)}1_{Y_{j+1}>Y_{j}}+1_{Y_{j+1}<Y_{j}}\bigg)
\Bigg)\,=\,\hfill\cr
\sum_{I\subset\{\,1,\dots,k-1\,\}}
\bigg(
\prod_{j\in I}
\displaystyle\frac{f(j)}{f(0)}
\bigg)P\big(\text{the ascents of $\fs^{-1}$ are $I$}\big)\,.
}$$
The last probability is equal to 
$\frac{1}{k!}\updown{k}{I}$.
We conclude that
$$
\displaylines{
\wE\bigg(\sum_{n=1}^\tau1_{\{S_n=k\}}\bigg)\,=\,
\frac{
a^{k-1}
}{k!}
\Bigg(
\prod_{i=1}^{k}
\frac{f(0)}{f(0)-f(i)}
\Bigg)
\sum_{I\subset\{\,1,\dots,k-1\,\}}
\updown{k}{I}
\prod_{j\in I}
\displaystyle\frac{f(j)}{f(0)}
\,.
}$$
Up to a multiplicative constant,
we recognize the formula obtained in section~\ref{upc} for the quantities $z(k)$.
\end{proof}
This last probabilistic construction matches the up--down formula for the quasispecies,
yet it is still mysterious. We have no convincing explanation so far for the presence
of the up--down coefficients. 
\bigskip

\noindent
{\bf Acknowledgements:} This work was completed during a visit to the mathematics department of the University of Padova. 
We warmly thank Carlo Mariconda and Paolo Dai Pra for their hospitality.
\bibliographystyle{plain}
\bibliography{sga}

\begin{thebibliography}{1}

\bibitem{CA}
L.~Carlitz.
\newblock Permutations with prescribed pattern.
\newblock {\em Math. Nachr.}, 58:31--53, 1973.

\bibitem{CD}
Rapha{\"e}l Cerf and Joseba Dalmau.
\newblock Quasispecies on class-dependent fitness landscapes.
\newblock {\em Bulletin of Mathematical Biology}, 78(6):1238--1258, 2016.

\bibitem{EI1}
Manfred Eigen.
\newblock Self-organization of matter and the evolution of biological
  macromolecules.
\newblock {\em Naturwissenschaften}, 58(10):465--523, 1971.

\bibitem{GKP}
Ronald~L. Graham, Donald~E. Knuth, and Oren Patashnik.
\newblock {\em Concrete mathematics}.
\newblock Addison-Wesley Publishing Company, Reading, MA, second edition, 1994.
\newblock A foundation for computer science.

\bibitem{MT}
Carlo Mariconda and Alberto Tonolo.
\newblock {\em Discrete calculus --Methods for counting--}.
\newblock Springer, 2016 (to appear).

\bibitem{OV}
G.~Oshanin and R.~Voituriez.
\newblock Random walk generated by random permutations of
  {$\{1,2,3,\dots,n+1\}$}.
\newblock {\em J. Phys. A}, 37(24):6221--6241, 2004.

\bibitem{SEN}
E.~Seneta.
\newblock {\em Nonnegative matrices and {M}arkov chains}.
\newblock Springer Series in Statistics. Springer-Verlag, New York, second
  edition, 1981.

\bibitem{Sloane}
N.~J.~A. Sloane.
\newblock The on-line encyclopedia of integer sequences.
\newblock {\em Ann. Math. Inform.}, 41:219--234, 2013.

\end{thebibliography}
 \thispagestyle{empty}

\end{document}